\documentclass[aps, pra, twocolumn, superscriptaddress, english]{revtex4-1}
\usepackage{amssymb,amsmath,amsthm}
\usepackage{hyperref}
\usepackage{enumerate}
\usepackage{verbatim}
\usepackage{amsmath}
\usepackage{natbib}
\usepackage{latexsym}
\usepackage{lipsum,babel}
\usepackage{yfonts}
\usepackage{amsfonts}
\usepackage{times}
\usepackage{graphicx}
\usepackage{bm,bbm,color,epstopdf}

\newtheorem{theorem}{Theorem}[section]
\newcommand{\rk}{{\rm rank}}
\newcommand{\id}{\mathbb{I}}

\newtheorem{lemma}[theorem]{Lemma}

\newtheorem{definition}[theorem]{Definition}
\newtheorem{remark}[theorem]{Remark}
\newtheorem{corollary}[theorem]{Corollary}
\newcommand{\specialcell}[2][c]{%
  \begin{tabular}[#1]{@{}c@{}}#2\end{tabular}}
  \usepackage{makecell}
\DeclareMathOperator{\diag}{diag}
\newcommand{\rr}{\mathbb{R}}
\DeclareMathOperator{\tr}{tr}

\newcommand{\pp}{\mathbb{P}}

\begin{document}

\title{Limits to catalysis in quantum thermodynamics}
\author{N. H.Y. \surname{Ng}}
\affiliation{Centre for Quantum Technologies, National University of Singapore, 117543 Singapore}
\author{L. \surname{Man\v{c}inska}}
\affiliation{Centre for Quantum Technologies, National University of Singapore, 117543 Singapore}
\author{C. \surname{Cirstoiu}}
\affiliation{Centre for Quantum Technologies, National University of Singapore, 117543 Singapore}
\affiliation{Department of Applied Mathematics and Theoretical Physics, University of Cambridge, Cambridge, CB3 0WB United Kingdom}
\author{J. \surname{Eisert}}
\affiliation{Dahlem Center for Complex Quantum Systems, Freie Universit\"{a}t Berlin, 14195 Berlin, Germany}
\author{S. \surname{Wehner}}
\affiliation{Centre for Quantum Technologies, National University of Singapore, 117543 Singapore}
\affiliation{School of Computing, National University of Singapore, 117417 Singapore}

\begin{abstract}
Quantum thermodynamics is a research field that aims at fleshing out the ultimate limits of thermodynamic
processes in the deep quantum regime. A complete picture of quantum thermodynamics allows for catalysts, i.e.,
systems facilitating state transformations while remaining essentially intact in their state,
very much reminding of catalysts in chemical reactions. In this work, we present a comprehensive analysis of the power and limitation
of such thermal catalysis. Specifically, we provide a family of optimal catalysts that can be returned with
minimal trace distance error after facilitating a state transformation process. To incorporate the genuine physical
role of a catalyst, we identify very significant restrictions on arbitrary state transformations under dimension or
mean energy bounds, using methods of convex relaxations. We discuss the implication of these findings on
possible thermodynamic state transformations in the quantum regime.
\end{abstract}

\maketitle

\section{Introduction}

In chemical reactions, it is common that a certain reaction should in principle be allowed, but in reality cannot take place (or occurs at extremely low rates) because of the presence of some large energy barrier. Fortunately, the situation is sometimes redeemed by the presence of certain chemical substances, referred to as catalysts, which effectively lower the energy barrier across the transformation. That is to say, catalysts significantly increase the reaction rates. Importantly, these catalysts can remain unchanged after the occurrence of the reaction, and hence a small amount of catalytic substance could be used repeatedly and is sufficient to facilitate the chemical reaction of interest.

The basic principles of chemical reactions are governed by thermodynamic considerations such as the second law. There have specifically been a number of recent advances in the quest of understanding the fundamental laws of thermodynamics \cite{secondlaw,HO03,BMORS13,SSP13,EDRV12,Gallego}. 
These efforts are especially focused on the quantum nano-regime, where finite size effects and quantum coherences are 
becoming increasingly relevant. One particularly insightful approach is to cast thermodynamics as a resource theory \cite{HO03,BMORS13, janzing2000, 1159774}, reminding of notions in entanglement theory \cite{Horodecki02, LTW11, LW13}. In this framework, thermodynamics can be seen as the theory that describes conditions for state transformation  $\rho\rightarrow\sigma$ 
from some quantum state to another
under \emph{thermal operations} (TO). The notion of TO means allowing for the full set of global unitaries which are energy preserving in the presence of some thermal bath. This is a healthy and fruitful standpoint, and allows the application of many 
concepts and powerful tools derived from information theory \cite{FDOR12,dahlsten2011, delrio2011}. 


In the context of thermal operations, catalysts emerge as ancillatory systems that facilitate state transformation processes: there are cases where $\rho\rightarrow\sigma$ is not possible, but there exists a state $\omega_C$ such that $\omega_C\otimes\rho\rightarrow\omega_C\otimes\sigma$ \emph{is} possible. The metaphor of catalysis is appropriate indeed:
This implies that by using such a catalyst $\omega_C$, one is enabled to perform the thermodynamic transformation $\rho\rightarrow\sigma$, while returning the catalyst back in its \emph{exact} original form. This is called \emph{exact catalysis}. The inclusion of catalyst states in thermal operations serve as an important step in an eventual complete picture of
quantum thermodynamics; it allows us to describe thermodynamic transformations in the full picture, where the system is interacting with experimental apparatus, for example a clock system. 
Furthermore, it has been shown that one can obtain necessary and sufficient conditions for exact catalysis 
in terms of a whole family of generalised free energies \cite{secondlaw}. The ordinary second law of ever-decreasing free energy is but the constraint on one of these free energies. 

Naturally, for physically realistic scenarios inexact catalysis is anticipated, where the catalyst is returned except for a slight degradation. However,
rather surprisingly, it has been shown \cite{secondlaw} that at least in some cases, the conditions for catalytic transformations are highly non-robust against small errors induced in the catalyst. 
The form of the second law thus depends crucially on the measure used to quantify inexactness.
In particular, if inexactness is defined in terms of small trace distance, then there is no second law at all: for any $\varepsilon>0$ one could pick \textit{any} two states $\rho$ and $\sigma$, and starting from $\omega_C\otimes\rho$, get $\varepsilon$-close in terms of trace distance to $\omega_C\otimes\sigma$ via thermal operations.
We refer to this effect as \emph{thermal embezzling}: Here one observes that instead of merely catalysing the reaction, energy/purity has possibly been extracted 
from the catalyst and used to facilitate thermodynamic transformations, while leaving the catalyst state arbitrarily close to being intact \cite{vDH03}. On physical grounds, such a 
setting seems implausible, even though it is formally legitimate. A clarification of this puzzle seems very much warranted.

Argued formally, a first hint towards a resolution may be provided by looking at how the error depends on the system size.
Naturally, the trace distance error $\varepsilon$ depends on the dimension of the catalyst states $\dim(\omega_C)=n$; nevertheless one can find examples of catalysts where $\varepsilon\rightarrow 0$ as $n$ approaches infinity.  While examples show that in principle thermal embezzling may occur \cite{secondlaw}, hardly anything else is known otherwise.
 Indeed, it would be interesting to understand the crucial properties that distinguish between a catalyst and an active reactant in thermodynamics.
From a physical perspective, it seems highly desirable to understand to what the effect of embezzling can even occur for physically plausible systems.

In this work, we highlight both the power and limitations of thermal catalysis, by providing comprehensive answers to the above questions raised. 
Firstly, we construct a family of catalyst states depending on dimension $n$, which achieves the optimal trace distance error while facilitating the state transformation $\rho\rightarrow\sigma$, for $\rho$ and $\sigma$ being some arbitrary $m$-dimensional states. This is done for the regime where the Hamiltonians of the system and catalyst are trivial. Secondly, we show that thermal embezzling with arbitrary precision cannot happen under reasonable constraints on the catalyst.
More precisely, whenever the dimension of the catalyst is bounded, we derive non-zero bounds on the trace distance error. By making use of splitting techniques to simplify the optimization problems of interest, such bounds can also be obtained when the expectation value of energy of the catalyst state is finite, for catalyst Hamiltonians with \emph{unbounded} energy eigenvalues and a finite partition function.
We hence set very strong limitations on the possibility of enlarging the set of allowed operations in quantum thermodynamics, if systems with reasonable Hamiltonians are being considered.


\section{Results}
\subsection{The power of thermal embezzling}
We begin by exploring the case for vanishing, trivial Hamiltonians, where it is known that thermal embezzling can occur. This is also the simplest case of thermodynamics in resource theory \cite{secondlaw}, when all energy levels are fully degenerate, and the Hamiltonian is simply proportional to the identity operator. Entropy and information, instead of energy, become the main quantity that measures the usefulness of resources. In such cases, the sole conditions governing a transition from some quantum state $\rho$ to $\sigma$ is that the eigenvalue vector of $\rho$ majorizes that of $\sigma$ \cite{HO03}. This is commonly denoted as $\rho\succ\sigma$. Such a condition also implies that entropy can never decrease under thermal operations~\cite{secondlaw}.

To investigate thermal embezzling in this setting, one asks if given fixed $m,n$, what is the smallest $\varepsilon$ such that there exists a catalyst state $\omega_C$ that satisfies 
\begin{equation}\label{eq:statetransf}
\omega_C\otimes\frac{\mathbb{I}}{m}\succ\omega_C'\otimes|0\rangle\langle 0|,
\end{equation}
where the trace distance $d(\omega_C,\omega'_C)$ between the input catalyst $\omega_C$ and output catalyst $\omega'_C$ is not greater than $\varepsilon$. This trace distance is used as a measure of catalytic error throughout our analysis. If some catalyst pair $(\omega_C,\omega'_C)$ satisfies condition Eq.~\eqref{eq:statetransf} with trace distance $\varepsilon$, then it also facilitates $\omega_C\otimes\rho\rightarrow\omega'_C\otimes\sigma$ for any $m$-dimensional states $\rho,\sigma$. This is because a pure state majorizes any other state, while the maximally mixed state ${\id}/{m}$ is majorized by any other state.

Since majorization conditions depend solely on the eigenvalues of the density matrices $\omega_C$ and $\omega'_C$, one can phrase this problem of state transformation in terms of a linear minimization program over catalyst states diagonal and ordered in the same basis (see appendix). In fact, the eigenvalues of $\omega_C, \omega'_C$ which give rise to optimal trace distance error can be solved by such a linear program, although these eigenvalues are non-unique. Whenever $m\geq 2$ and $n=m^a$ where $a\geq 1$ is an integer, we provide an analytic construction of catalyst states, which we later show to be optimal for the state transformation in Eq.~\eqref{eq:statetransf}. Let the initial catalyst state $\omega_C = \sum_{i=1}^n \omega_i |i\rangle\langle i|$, where $\omega_1 = {m}/({1+(m-1)a})$, 
\begin{equation}
  \omega_i = \begin{cases}
     \omega_1  m^{-  \left \lceil{\log_m i}\right \rceil }& \text{if $2 \leq i \leq n/m$},\\
     0 & \text{if $i>n/m$}.
  \end{cases}
\end{equation}
Note that our catalyst state $\omega_C$ does not have full rank, and this is crucial for the majorization condition in Eq.~\eqref{eq:statetransf} to hold, since $\rho\succ\sigma$ implies that $\rk(\rho)\leq\rk(\sigma)$, and the joint state $\omega_C'\otimes|0\rangle\langle 0|$ can have at most rank $n$. The output catalyst $\omega_C'$ can be obtained from $\omega_C$, by subtracting a small value $\varepsilon$ from the largest eigenvalue $\omega_1$ and distributing the amount $\varepsilon$ equally over the indices $i>n/m$. This brings $\omega_C'$ to be a state of full rank $n$. We show that this family achieves trace distance error
\begin{equation}\label{eq:opt_trdist}
d_{m,n} = \frac{m-1}{1+(m-1)\log_m n},
\end{equation}
which we prove by mathematical induction to be optimal, given fixed $m, n$ where $n=m^a$ (see appendix).
\begin{figure}
    \includegraphics[width=0.4\textwidth]{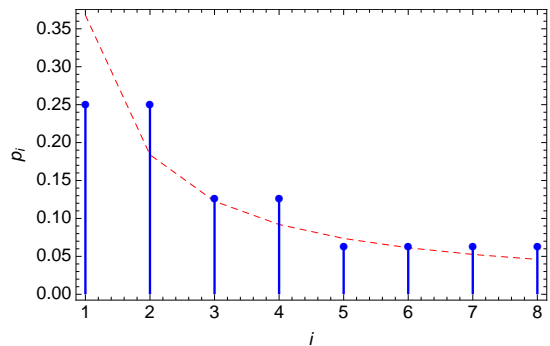}
    \caption{The eigenvalues of our final catalyst state $\omega_C'$ (blue) versus those of $\tilde{\omega}_C$ proposed in Ref. \cite{vDH03} (red, dashed), for $m=2$ and $n=8$.}
    \label{fig:1}
\end{figure}
\begin{figure}[h!]
    \includegraphics[width=0.4\textwidth]{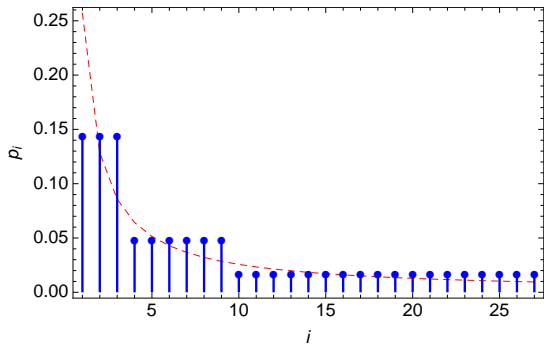}
    \caption{The eigenvalues of our final catalyst state $\omega_C'$ (blue) versus those of $\tilde{\omega}_C$ proposed in Ref. \cite{vDH03} (red, dashed), for $m=3$ and $n=27$.}
    \label{fig:2}
\end{figure}

Figs.\ \ref{fig:1} and \ref{fig:2} compare our final catalyst state with the state 
\begin{equation}
\tilde{\omega}_C = \frac{1}{C(n)} \sum_{i=1}^n \frac{1}{i} |i\rangle\langle i|,
\end{equation}
with $C(n)= \sum_{i=1}^n {1}/{i}$ being the normalization constant. The family $\tilde{\omega}_C$ was proposed in Ref.~\cite{vDH03} for embezzling in the LOCC setting. In Fig.\ \ref{fig:3}, we compare the trace distance error achieved by catalyst $\tilde{\omega}_C$ from Ref.~\cite{vDH03} with the error achieved by our catalyst $\omega_C$. We see that for small dimensions, our catalyst outperforms $\tilde{\omega}_C$, however asymptotically the error scales with $\log n$ for both catalysts.

\begin{figure}
\includegraphics[scale=0.65]{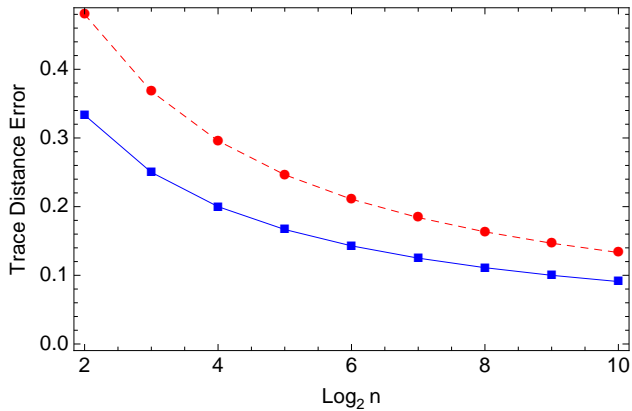}
\caption{The comparison of trace distance error for our state (blue, solid) and the catalyst state in Figs.\ \ref{fig:1} and \ref{fig:2} (red, dashed), for the case where $m=2$.}
\label{fig:3}
\end{figure}

\subsection{The limits of thermal embezzling}
In this section, we are interested in finding additional physical restrictions which prevent thermal embezzling. To do so, we look at general Hamiltonians for both the system and catalyst, where the energy of the system comes into play. 
In \cite{secondlaw}, it is shown that the monotonicity of quantum R{\'e}nyi divergences \cite{muller2013quantum} form necessary conditions for state transformations. More, precisely, for arbitrary $\rho_S$ and $\rho_S'$, if $\rho_S\rightarrow\rho_S'$ is possible via catalytic thermal operations, then for all $\alpha\geq 0$, 
\begin{equation}\label{eq:secondlaw}
D_\alpha(\rho_S\|\tau_S)\geq D_\alpha(\rho_S'\|\tau_S)
\end{equation} 
holds, where $\tau_S$ is the thermal state of system $S$, at temperature $T$ of the thermal bath. 

Eq.~\eqref{eq:secondlaw} implies that one can use the monotonicity of R{\'e}nyi divergences to find lower bounds on thermal embezzling error for state transformation between arbitrary states $\rho_S$ and $\rho_S'$. For simplicity, we present the case where $\rho_S$ and $\rho_S'$ are diagonal (in the energy eigenbasis of $H_S$). The case for arbitrary states can be treated similarly, and details are given in the appendix.

For the case where two states $\rho$ and $\sigma$ are diagonal, the R{\'e}nyi divergences are defined as 
\begin{equation}
D_\alpha(\rho\|\sigma):= \frac{1}{\alpha-1}\log \sum_{i} \rho_i^\alpha \sigma_i^{1-\alpha},
\end{equation}
where $\lbrace \rho_i\rbrace,\lbrace \sigma_i\rbrace$ are the eigenvalues of $\rho$, and $\sigma$. 

Again, for states $\rho_S$ and $\rho_S'$ diagonal, it suffices to look at a single transformation
\begin{equation}\label{eq:univ2}
\omega_C\otimes\tau_S \rightarrow \omega_C'\otimes \Pi_{\max}^S,
\end{equation} 
where $\Pi_{\max}^S=|E^S_{\rm max}\rangle\langle E^S_{\rm max}|$ is the pure energy eigenstate with energy $E^S_{\rm max}$. 
Note that both $\tau_S$ and $\Pi_{\max}^S$ are diagonal in the energy eigenbasis. As explained in the appendix, Eq.~\eqref{eq:univ2} is sufficient to ensure universal thermal embezzling for aribtrary states $\rho_S$ and $\rho_S'$ as long as they are diagonal in the same energy eigenbasis. Similarly, one can take $\omega_C$ and $\omega'_C$ to be diagonal in the energy eigenbasis of $H_C$~\cite{secondlaw}.
This can be written as the following minimization problem, $\varepsilon$ being the solution of
\begin{eqnarray}\label{array:min1}
{\rm min} &&\quad\frac{1}{2}\|\omega_C-\omega_C'\|_1 \\
{\rm s.t.} &&\quad \forall\alpha\geq 0, D_\alpha(\omega_C\otimes\tau_S\|\tau_{CS})\geq D_\alpha(\omega_C'\otimes\Pi_{\max}^S\|\tau_{CS}),\nonumber\\
&&\quad 0\leq \omega_C, \omega_C' \leq \id,\nonumber
\end{eqnarray}
where $\tau_{CS}=\tau_C\otimes\tau_S$ is the thermal state of the catalyst and system.
The system Hamiltonian ${H}_S$ is assumed to be finite.

A straightforward relaxation of Eq.~\eqref{array:min1} allows us to now consider an alternative problem for some fixed $\alpha$
\begin{align}\label{array:min2}
\varepsilon_\alpha := \min \quad& \frac{1}{2}\|\omega_C-\omega_C'\|_1 \\
{\rm s.t.} \quad& D_\alpha (\omega_C\otimes\tau_S\|\tau_{CS}) \geq D_\alpha(\omega_C'\otimes \Pi_{\max}^S \|\tau_{CS}),\nonumber\\
\quad & 0\leq \omega_C, \omega_C' \leq \id .
\nonumber
\end{align}
From Ref.~\cite{secondlaw}, we know that any $(\omega_C,\omega_C')$ feasible for Eq.~\eqref{array:min1} is also feasible for Eq.~\eqref{array:min2}. Therefore, for any $\alpha\geq 0$, $\varepsilon\geq\varepsilon_\alpha$.
By choosing $\alpha$ one can arrive at much simpler optimization problems, that provide lower bounds for the trace distance error. We apply this to study two cases, detailed as below.

{\bf 1. Bounded dimension:} Consider the case where both the system $H_S$ and catalyst Hamiltonians ${H}_C$ have fixed dimensions, and denote the maximum energy eigenvalues as $E_{\rm max}^S, E_{\rm max}^C$ respectively. One sees that the solution of Eq.~\eqref{array:min1} is lower bounded by Eq.~\eqref{array:min2} for $\alpha\rightarrow\infty$. Recall that w.l.o.g.\ we can assume that $\omega_C$ and
$\omega_C'$  are diagonal in the same basis, which we take to be the energy eigenbasis. Eq.~\eqref{array:min2} can be rewritten as
\begin{eqnarray}
\varepsilon_\infty=\min &&\quad \frac{1}{2}\sum_i |\omega_i-\omega_i'|\\
{\rm s.t.} &&\quad \max_i \frac{\omega_i}{\tau_i} \geq \frac{Z_S}{e^{-\beta E_{\rm max}^S}} \max_i \frac{\omega_i'}{\tau_i},
\nonumber\\
&&\quad \sum_i \omega_i = \sum_i \omega_i' = 1,\nonumber
\end{eqnarray}
where $\tau_i = {Z_C}^{-1} e^{-\beta E^C_i}$ are the probabilities defined by the thermal state of the catalyst Hamiltonian, and $Z_C$ is the partition function of the catalyst system. To solve this problem, we note that the optimal strategy to maximize the quantity $\max_i {\omega_i}/{\tau_i}$ within the $\varepsilon-$ball of $\omega'_C$ is to increase one of the eigenvalues by $\varepsilon$, so that the quantity 
$\max_i ({\omega_i+\varepsilon})/{\tau_i}$ is maximized. With further details in the appendix, we show that the trace distance error can therefore be lower-bounded by 
\begin{equation}\label{eq:opt_BD}
d_{\rm opt} ({H}_S, {H}_C) \geq \left(\frac{Z_S}{e^{-\beta E_{\rm max}^S}} -1\right) \frac{e^{-\beta E_{\rm max}^C}}{Z_C},
\end{equation}
where $Z_S, Z_C$ are the partition functions of the system and catalyst.

Although this bound is valid for arbitrary finite-dimensional Hamiltonians, it is not tight. Indeed, in the case of trivial Hamiltonians where all states have constant energy value, normalized to 0, the partition functions $Z_S,Z_C$ reduce to the dimension $m,n$ of the system and catalyst. This bound then yields $d_{\rm opt}(\mathbf{0}_S,\mathbf{0}_C) \geq (m-1)/{n}$, which is much weaker than the optimal trace distance we derived in Eq.~\eqref{eq:opt_trdist}.


%
%

{\bf 2. Hamiltonians with unbounded energy levels:} A more general result holds for unbounded dimension and energy levels where the partition function $Z_C$ is finite. More precisely, for such cases, we show that setting an upper bound on the average energy of the catalyst state limits thermal embezzling.

Let us now explain the proof of our results. Consider some ${H}_C$ with unbounded energy levels $\lbrace E_j^C \rbrace$. For simplicity, we restrict ourselves to the case where the catalyst states are diagonal in the energy eigenbasis, and assume the system Hamiltonian to be trivial with dimension $m=2$. A more general derivation involving arbitrary system Hamiltonians may be found in the appendix. \vspace{0.2cm}
\\
{\it A) Formulation of the problem:} 
Consider the minimization of catalytic error under the relaxed constraint that monotonicity for the $\alpha$-R{\'e}nyi divergence is satisfied. Using Eq.~\eqref{array:min2} with $\alpha={1}/{2}$, by substituting ${H}_S=\mathbf{0}_S$, the first constraint can be simplified as follows
\begin{eqnarray*}
	D_{1/2} \Bigl(\omega_C\otimes\frac{\id}{2}\|\tau_{C}\otimes\frac{\id}{2}\Bigr) &\geq& D_{1/2} \Bigl(\omega_C'\otimes|0\rangle\langle 0| \|\tau_{C}\otimes\frac{\id}{2}\Bigr),\\
 D_{1/2} (\omega_C\|\tau_C) &\geq& D_{1/2} (\omega_C'\|\tau_C) + D_{1/2} \Bigl(|0\rangle\langle 0|\|\frac{\id}{2}\Bigr),\\
\sqrt{2} \sum_{i} {\omega_i}^{1/2} e^{-\beta E_i^C/2} &\leq&  \sum_i {\omega_i'}^{1/2}
e^{-\beta E_i^C/2}.
 \end{eqnarray*}
Furthermore, we want that the initial catalyst state must have a expectation value of energy no larger than some finite $E$. 
In summary, we now look at the minimization of trace distance under the following constraints
\begin{eqnarray}\label{array:min3}
{\rm min} &&\quad \sum_i |\omega_i - \omega_i'| \\
{\rm s.t.} &&\quad \sum_i \sqrt{\omega_i'} \gamma^{E_i^C}\geq \sqrt{2}\sum_i \sqrt{\omega_i}\gamma^{E_i^C},\nonumber\\
&&\quad \omega_i, \omega_i'\geq 0 \,\,\,\forall i, ~\sum_{i} E_i^C \omega_i \leq E,\nonumber
\end{eqnarray}
where $\gamma=e^{-\beta/2 }\in (0,1)$. Denote the solution of this problem as $\varepsilon$. In the subsequent steps, our goal is to show that $\varepsilon$ is lower bounded by a non-zero constant, by making use of techniques of convex relaxations of optimisation problems. As such, this is an intricate problem, as it is a non-convex problem both in $\omega_i$ and $\omega_i'$.

{\it B) Splitting a relaxed minimization problem:}
The key idea to proceed is to suitably split the problem into two independent optimization problems in a relaxation,
which then turn out to be convex optimization problems the duals of which can be readily assessed. The starting point of this approach is rooted in the observation
that for any $\omega_i,\omega_i'\in [0,1]$, the following inequality holds true,
\begin{equation}\label{eq:split}
{\omega_i'}^{1/2}-{2\omega_i}^{1/2} \leq {|\omega_i'-\omega_i|}^{1/2} - {\omega_i}/{3}.
\end{equation}
Since requiring the R.H.S. of Eq.~\eqref{eq:split} to be positive is a less stringent compared to the L.H.S., one can now further use it to obtain a lower bound for the minimization in Eq.~\eqref{array:min3}. By defining a new variable $x_i=|\omega_i-\omega_i'|$, we arrive at a new problem
\begin{eqnarray}\label{array:min4}
{\rm min} &&\quad \sum_i x_i \\
{\rm s.t.} &&\quad \sum_i \sqrt{x_i} \gamma^{E_i^C}\geq \frac{1}{3}\sum_i \sqrt{\omega_i}\gamma^{E_i^C},\nonumber\\
&&\quad x_i, \omega_i\geq 0 \,\,\,\forall i, ~\sum_{i} E_i^C \omega_i \leq E ,\nonumber
\end{eqnarray}
the solution $\zeta$ of which obeys $\varepsilon\geq\zeta$. In the next step, we will see that the relaxed problem in Eq.\eqref{array:min4} is much simpler to solve, since it can be written as two separate, independent optimization problems. One can see now that the variables $x_i, \omega_i$ are independent from each other. This allows us to first perform a minimization of the function $\sum_i \sqrt{\omega_i}\gamma^{E_i^C}$ for constraints involving $\omega_i$ only.

{\it C) Invoking energy constraints to provide lower bound:}
The energy constraint on $\omega_C$ plays a crucial role in lower bounding the solution. Intuitively, when such a constraint is placed for some finite $E$, it implies that the probability of populating some relatively low energy levels \textit{cannot} be vanishingly small. We prove this with more rigor in the appendix.
Along this line of reasoning, one concludes that for the minimization 
\begin{eqnarray}\label{array:min5}
{\rm min} &&\quad  \sum_i \sqrt{\omega_i}\gamma^{E_i^C} \\
{\rm s.t.} &&\quad \omega_i\geq 0 \,\,\,\forall i, ~\sum_{i} E_i^C \omega_i \leq E ,\nonumber
\end{eqnarray}
its solution $\varepsilon_1 >0$ has to be strictly positive. More precisely,
\begin{eqnarray}
\varepsilon_1~&&=~ \displaystyle\max_{W\in(0,1)} W  \gamma^{E_j^C(W)},
\end{eqnarray}
where $j(W)= \min \lbrace j: E_{j+1}^C> {E}/({1-W})$. 
A derivation of this expression can be found in the appendix.\vspace{0.2cm}
\\
{\it D) Merging both problems:}
After obtaining a lower bound for the subproblem Eq.~\eqref{array:min5}, we recombine the two problems into Eq.~\eqref{array:min4} to obtain
\begin{eqnarray}\label{array:min6}
{\rm min} &&\quad \sum_i x_i \\
{\rm s.t.} &&\quad \sum_i \sqrt{x_i} \gamma^{E_i^C}\geq \frac{1}{3}\varepsilon_1, \nonumber\\
&&\quad x_i \geq 0 \,\,\,\forall i. \nonumber
\end{eqnarray}
This is a quadratic optimization problem in the variables $\sqrt{x_i}$, hence it is easy to obtain the Lagrange dual of this problem, which takes on a very simple form
\begin{eqnarray}\label{array:min7}
{\rm min} &&\quad - \frac{1}{4} \lambda^2 \sum_i \gamma^{2E_j^C} +\lambda\varepsilon_1 \\
{\rm s.t.} &&\quad \lambda \geq 0 \nonumber,
\end{eqnarray}
involving the simple minimization of a quadratic function w.r.t. $\lambda$. Solving this we arrive at a lower bound
\begin{equation}\label{eq:ELB}
\varepsilon \geq \zeta\geq \frac{1}{9}\frac{\varepsilon_1^2}{Z_C} >0,
\end{equation}
where $Z_C=\sum_i \gamma^{2E_i^C} = \sum_i e^{-\beta E_i^C}$ is the partition function of the catalyst Hamiltonian. 
We summarize our findings in Table \ref{table:comp}.

\begin{table}[h]
\begin{tabular}{|c|c|c|}\hline
\diaghead{\theadfont Diag ColumnmnHead II}%
{Energy levels}{Dimension}&
Bounded&Unbounded\\
\hline 
Fully degenerate & No & Yes \\ 
\hline 
Bounded & No & \specialcell{Probably, true at least for \\fully degenerate Hamiltonians} \\ 
\hline 
\specialcell{Unbounded} & N/A & \specialcell{No, if average energy and part- \\ition function is finite} \\ 
\hline 
\end{tabular}
\caption{The occurance of thermal embezzling with arbitrary precision, under different settings. For regimes labeled "No", explicit bounds can be found in Eq.~\eqref{eq:opt_trdist}, Eq.~\eqref{eq:opt_BD} and Eq.~\eqref{eq:ELB}.}
\label{table:comp}
\end{table}

\section{Discussion and conclusion}
The bounds on dimensionality are closely related to energy restrictions. While placing an upper bound on the dimension directly imply an upper bound on the average energy, the reverse statement is not generally true. However, if one restricts not only the expectation value of the energy distribution, but also restricts its variance to be finite, then this is almost equivalent to placing a dimension restriction. For example, given any non-degenerate Hamiltonian ${H}_C$ with unbounded eigenvalues, consider the set of catalyst states such that the average energy and variance of a given catalyst is finite. Then by the Chebyshev inequality one can understand that this is equivalent to introducing a cut-off on the maximum energy eigenvalue (and therefore on the dimension). We note that it is easy to see that e.g. for the harmonic oscillator the variance is not always bounded whenever the mean energy is bounded.

In the case of infinite-dimensional Hamiltonians, we have also shown that for certain classes of catalyst Hamiltonians, explicit bounds can be derived on the trace distance error of a catalyst when the average energy is finite. Our results have covered a large range of Hamiltonians which are commonly found in physical systems - including the important case of the Harmonic oscillator in free systems - with the minimal assumption that partition function $Z_C$ is finite, which holds for all systems for which the canonical ensemble is well-defined. However, we know that thermal embezzling can be arbitrarily accurate as dimension grows, at least in the simplest case of the trivial Hamiltonian. This implies that there will be specific cases of infinite-dimensional Hamiltonians where simply bounds on average energy do not give explicit bounds on thermal embezzling error. We suspect that this may be true for Hamiltonians with unbounded dimension, but upper bounded energy levels. The reason is that if dimension is unbounded, then there must exist an accumulation point in the energy spectrum. The subspace of this accumulation point will be very similar to the trivial Hamiltonian.

In summary, we have investigated the phenomenon of thermal embezzling under different physical scenarios. While one acknowledges that thermal embezzling is possible in the fully degenerate Hamiltonian case, we show that under many realistic circumstances, with physically motivated restrictions, thermal embezzling cannot happen with arbitrary accuracy. 
In this sense, we resolve the puzzle of thermal embezzling, hence further contributing to a complete understanding of the thermodynamic laws in the quantum world.

{\it Acknowledgements:} This work has been supported by the EU (RAQUEL, SIQS), the ERC (TAQ), the FQXi, and the COST network. NN, LM, CC and SW are funded by the Ministry of Education (MOE) and National Research Foundation Singapore, as well as MOE Tier 3 Grant "Random numbers from quantum processes" (MOE2012-T3-1-009).

\vspace{8cm}
\onecolumngrid
\section*{Appendix}
\setcounter{equation}{0}
\setcounter{section}{0}
In this appendix we fully elaborate our findings on thermal catalysis. We begin in Section \ref{subsec:resource} by explaining the similarities and subtle differences between thermal embezzling and embezzling in the LOCC setting. The R{\'e}nyi divergences and their relation to thermal operations are detailed in Section \ref{subsec:d_alpha}. Proceeding to Section \ref{sec:OptFamilyL}, we focus on thermal embezzling for trivial Hamiltonians with fixed dimensions. On the one hand, we investigate the problem of finding a catalyst which allows us to perform thermal embezzling with minimum possible error in trace distance.
We detail the proofs on our construction of a catalyst family (given dimension parameters for both system of interest and catalyst), and prove that our construction achieves the optimal embezzling error. 

On the other, by placing restrictions on the dimension, we derive non-zero lower bounds for embezzling error, considering arbitrary system and catalyst Hamiltonians. The proofs are detailed in Section \ref{sec:DimBounds}. Some technical background on $\alpha-$R{\'e}nyi divergences and their relation to thermodynamic operations are given. Lastly, in Section \ref{sec:ELB} we focus on infinite-dimensional Hamiltonians, with unbounded energy levels (and finite partition function). We show that as long as the average energy of the catalyst is finite, explicit lower bounds on accuracy of embezzling can be obtained.

\section{Embezzling and catalysis}\label{sec:emb&cat}
\subsection{Thermodynamics as a resource theory}\label{subsec:resource}
Resource theories are frameworks useful in identifying states which are valuable, under specific classes of allowed operations and states given for free. A state is a valuable resource if one can use it to create many other states under the set of allowed operations. Thermodynamics can be viewed as a resource theory \cite{HO03,BMORS13}, where the allowed operations are the so-called \textit{thermal operations}. They are summarized as follows: consider a system $S$, given a state $\rho_S$ and the Hamiltonian $H_S$, one can
\begin{enumerate}
\item for any bath system $B$ with Hamiltonian $H_B$, attach any thermal state $\tau_B = e^{-\beta H_B}/{\tr [e^{-\beta H_B}]}$ to $\rho_S$, where $\beta=1/{kT}$ is the inverse temperature,
\item perform any unitary $U$ over the global system $SB$ that commutes with the total Hamiltonian, i.e. $[U,H_S+H_B]=0$,
\item trace out the bath $B$.
\end{enumerate}


Recently, the framework of thermal operations have been used to prove a second law in Ref.\ \cite{secondlaw} by including catalytic effects. This is because there exists certain states $\rho$ and $\sigma$ such that via thermal operations, $\rho\nrightarrow\sigma$, but $\rho\otimes\omega_C\rightarrow\sigma\otimes\omega_C$ for some state $\omega_C$. More precisely, catalytic effects can be accounted for by adding a fourth rule 

\begin{enumerate}
\setcounter{enumi}{3}
\item for any catalyst system $C$ with Hamiltonian $H_C$, attach any additional catalyst state $\omega_C$,  as long as the returned state $\omega_C'$ is $\varepsilon-$close to its original state $\omega_C$,
\end{enumerate} 
to the set of allowed operations.
One can now ask, given $\rho_S$, what are the states $\rho_S'$ can be obtained from $\rho_S$ under approximate \emph{catalytic thermal operations}? More precisely, do there exist 
$\omega_C,\omega_C'$ which are $\varepsilon-$close to each other, such that $\omega_C\otimes\rho_S\rightarrow\omega_C'\otimes\rho_S'$? 

Depending on $\varepsilon$ and the measure of closeness used, the conditions for $\rho_S\rightarrow\rho_S'$ to occur can vary. For example, if $\varepsilon$ is required to be zero, i.e. the catalyst must be returned in its exact form, then Ref.\ \cite{secondlaw} shows for any $\rho_S$ and $\rho_S'$ such that $\rho_S\rightarrow\rho_S'$ is possible via catalytic thermal operations, a whole set of R{\'e}nyi divergences must necessarily decrease. In the next Section \ref{subsec:d_alpha}, we define the R{\'e}nyi divergences and state the results of \cite{secondlaw} in detail.
On the other hand, if $\varepsilon$ is measured in terms of trace distance between the input and output catalyst only, Ref.\ \cite{secondlaw} also proves that for any $\varepsilon>0$, the state transformation conditions are trivial, i.e. any $\rho_S$ can be transformed to any $\rho_S'$.
We denote thermal embezzling as the phenomenon that by requiring only the input and output catalyst to be close in terms of trace distance, one can achieve $\rho_S\rightarrow\rho_S'$ for any $\rho_S,\rho_S'$. 

Another well-studied example of a resource theory is entanglement theory, where the allowed operations are those that can be implemented using local operations and classical communiaction (LOCC), while free states are the set of separable states. The interconversion of resources states in entanglement theory have been studied intensively, and have also provided insight into the resource theory of thermodynamics.


Embezzling states were originally introduced for the LOCC setting in Ref.\ \cite{vDH03}. An entangled state $|\nu(n)\rangle_{AB}\in\mathcal{C}^n\otimes\mathcal{C}^n$ shared between two parties $A$ and $B$ can be used as a resource to prepare some other state (of much smaller dimension), 
\begin{equation}
|\nu(n)\rangle_{AB}\rightarrow_{1-\varepsilon} |\nu(n)\rangle_{AB} |\psi\rangle_{AB}.
\end{equation}
The fidelity between the actual final state with $|\nu(n)\rangle_{AB} |\psi\rangle_{AB}$ is denoted by $1-\varepsilon$ , which goes to 1 when $n$ goes to infinity. This enables the approximate preparation of the state $|\psi\rangle_{AB}$, while the embezzling resource state is also left close to its original state. Such a preparation can even be achieved simply via local operations (LO). The family $|\nu(n)\rangle_{AB}$ is called a \emph{universal embezzling state} if it enables the preparation of any $|\psi\rangle_{AB}$. While this seemingly violates entanglement monotonicity under LOCC operations, one quickly realises that it is really because the closeness in entanglement content of two states depend not only on the fidelity, but also the dimension. Hence entanglement is exhausted to prepare $|\psi\rangle_{AB}$, while $|\nu\rangle_{AB}$ remains close to intact on the whole. However, there is also something special about embezzling states, in the sense that a maximally entangled state does not serve as a good embezzling state. In Ref.\ \cite{LW13}, a comprehensive study about general characteristics of embezzling states was conducted, providing insight into the necessary structure of a state to be a good embezzler.
The power of embezzling in LOCC has been applied in several areas of quantum information, such as coherent state exchange protocols \cite{LTW11}, projection games \cite{DSV13}, or as a theoretical tool in proving the Quantum Reverse Shannon Theorem \cite{BCR11}.

There are some similarities between thermal embezzling and LOCC embezzling, however also many distinctive features exist. Most significantly, in thermodynamical systems, the Hamiltonian which determines the evolution of the system plays an important role in state conversion conditions \cite{HO13}. This feature is absent in LOCC embezzling. We summarize the similarities and differences of LOCC and thermal embezzling in Table \ref{supptab:1}.
\begin{table}[h!]
\begin{tabular}{|c|c|c|}
\hline 
 & \textbf{LOCC embezzling} & \textbf{Thermal embezzling} \\ 
\hline 
State conversion conditions & \multicolumn{2}{c|}{Related to majorization} \\ 
\hline 
Phenomena & \multicolumn{2}{c|}{\specialcell{The usage of a catalyst state of large dimension/energy while tolerating slight \\degradation allows the preparation of any desired target state to arbitrary precision}} \\ 
\hline 
Hamiltonians & Not of interest & Of much physical significance \\ 
\hline 
States & Pure, multipartite states & Mixed states in general \\ 
\hline 
\specialcell{Commonly used \\measure of closeness} & \specialcell{Fidelity of global state \\(system and embezzling state)} & \specialcell{Trace distance between \\input and output catalyst state} \\ 
\hline 
Allowed operations & LOCC/LO operations & Catalytic thermal operations \\
\hline
Accuracy limited by & Dimension & Dimension and energy \\
\hline
\end{tabular}\caption{An overview of differences between LOCC and thermal embezzling.}\label{supptab:1}
\end{table}

\subsection{R{\'e}nyi divergences as thermal monotones}
\label{subsec:d_alpha}
In this section we detail the conditions for state transformation under catalytic thermal operations, which are closely related to the R{\'e}nyi divergences. 
The simplest case of catalytic thermal operations is when all Hamiltonians $H_S, H_C$ are trivial. For arbitrary states $\rho$ and $\sigma$, $\rho\rightarrow\sigma$ is possible if and only if 
$\rho\succ\sigma$ \cite{HO03}.
In the case where $H_S$ or $H_C$ are generally non-trivial, state conversion conditions are affected by the involved Hamiltonians. More precisely, instead of majorization, we need to consider the monotonicity of R{\'e}nyi divergences as a (necessary) condition for state transformations. These conditions are used later in Sections \ref{sec:DimBounds} and\ref{sec:ELB} to investigate the limits of thermal embezzling. Let us first define these quantities in Definition \ref{def:D_alpha}.

\begin{definition}[R{\'e}nyi divergences \cite{muller2013quantum}]\label{def:D_alpha}
Given arbitrary states $\rho, \sigma \geq 0$, for $\alpha \in [0, \infty]$, the R{\'e}nyi divergence of $\rho$ relative to $\sigma$ is defined as
\begin{equation}
D_\alpha (\rho\|\sigma):= \frac{1}{\alpha-1}\log \left[\tr \left(\rho^\frac{1-\alpha}{2\alpha}\sigma\rho^\frac{1-\alpha}{2\alpha}\right)^\alpha\right]
\end{equation}
For $\rho,\sigma$ diagonal in the same basis, let $p=( p_1, p_2, ..., p_n)$ and $q=( q_1, q_2, ..., q_n)$ denote the eigenvalue vectors of the $\rho,\sigma$ respectively. Then the R{\'e}nyi divergences reduce to the form
\begin{equation}
D_\alpha(\rho\|\sigma) = D_\alpha(p \| q):= \frac{1}{\alpha-1} \log\sum_i^n p_i^\alpha q_i^{1-\alpha}.
\end{equation}
\end{definition}

It has been shown that the quantities $D_\alpha(\rho\|\tau)$ are thermal monotones, where $\tau$ is the thermal state of the system of interest. Intuitively, this implies that thermal operations can only bring the system of interest closer to its thermal state with the same temperature $T$ as the bath. We detail this in the following Lemma \ref{lem:secondlaw}.

\begin{lemma}[Monotonicity under thermal operations \cite{secondlaw}]\label{lem:secondlaw}
Given some Hamiltonian $H_A$, consider arbitrary states $\rho_A, \rho_A'$, where $\rho_A\rightarrow\rho_A'$ is possible via catalytic thermal operations. Denote by $\tau_A$ the thermal state of system $A$. Then for any $\alpha\in [0,\infty)$, 
\begin{equation}\label{eq:lemsecondlaw}
D_\alpha (\rho_A\|\tau_A) \geq D_\alpha (\rho_A'\|\tau_A).
\end{equation}
Furthermore, for any $\rho_A, \rho_A'$ diagonal in $H_A$, if Eq.~\eqref{eq:lemsecondlaw} holds for all $\alpha\geq 0$, then $\rho_A\rightarrow\rho_A'$ is possible via catalytic thermal operations.
\end{lemma} 

In essence, Lemma \ref{eq:lemsecondlaw} implies that the monotonicity of R{\'e}nyi divergences are necessary conditions for arbitrary state transformation, and for the case of states diagonal (in the energy eigenbasis), they are also sufficient. 
Let us also use a notation which was introduced in \cite{HO13} for diagonal states: we say that there exists a catalyst $\omega$ such that $\omega_C\otimes\rho_S \succ_T \omega_C\otimes\rho_S'$, if $\rho\rightarrow\sigma$ via catalytic thermal operations. We refer to the notion $\succ_T$ as thermo-majorization.

Now, let us consider the scenario of preparing a pure excited state of maximum energy $\Pi_{\max}^S=|E^S_{\rm max}\rangle\langle E^S_{\rm max}|$ from a thermal state $\tau_S$. Intuitively, if we concern ourselves only with diagonal state transformations, then this is the hardest thermal embezzling scenario possible. This is because for $\Pi_{\max}^S\succ_T\rho_S\succ_T\tau$ is possible for any diagonal $\rho_S$. Therefore, whenever we investigate the case where involved states are diagonal, it suffices to analyse the preparation of such a pure excited state. The necessary and sufficient conditions are
\begin{equation}\label{eq:thermalemb}
\omega_C\otimes\tau_S \succ_T \omega_C'\otimes\Pi_{\max}^S.
\end{equation}
In the next lemma, we show that given fixed Hamiltonians and dimensions, any catalyst state that succeeds in preparing such a state can also be used to facilitate any other state transformation.

\begin{lemma}[Universal embezzlers for diagonal states]\label{lem:univ_emb}
Suppose there exists $\omega_C,~\omega_C'$ diagonal (in $H_C$) such that $\omega_C\otimes\tau_S \succ_T \omega_C'\otimes\Pi_{\max}^S$ holds, and 
$\|\omega_C-\omega_C'\|_1=\varepsilon$. Then for any states $\rho_S, \rho_S'$ diagonal (in $H_S$), $\omega_C\otimes\rho_S \succ_T \omega_C'\otimes\rho_S'$ holds as well.
\end{lemma}

\begin{proof}
This can be proven by noting that 
\begin{equation}
	\omega_C\otimes\tau_S \succ_T \omega_C'\otimes\Pi_{\max}^S 
\end{equation}
is equivalent with the existence a thermal operation denoted by $\mathcal{M}$, such that $\mathcal{M}(\omega_C\otimes\tau_S)=\omega_C'\otimes\Pi_{\max}^S$. It remains to show that for any $\rho_S, \rho_S'$, there exists a thermal operation $\mathcal{M}'$ such that $\mathcal{M}' (\omega_C\otimes\rho_S) = \omega_C'\otimes\rho_S'$. 
Since the thermal state $\rho_S\succ_T\tau_S$ is thermo-majorized by any state $\rho_S$, and $\Pi_{\max}^S\succ_T \rho_S'$ thermo-majorizes any other state $\rho_S'$, there exist thermal operations $\mathcal{N}_1, \mathcal{N}_2$ such that $\mathcal{N}_1 (\rho_S) = \tau_S $ and $\mathcal{N}_2 (\rho_S) = \Pi_{\max}$.
Finally, consider 
\begin{equation}
	\mathcal{M'}=(\id_C\otimes\mathcal{N}_2) \circ \mathcal{M}\circ(\id_C\otimes\mathcal{N}_1), 
\end{equation}
then one sees that $\mathcal{M'} (\omega_C\otimes\rho_S) = \omega_C'\otimes\rho_S'$. This implies that $\omega_C\otimes\rho_S\succ_T\omega_C'\otimes\rho_S'$.
\end{proof}

\section{Optimal thermal catalyst for trivial Hamiltonians}\label{sec:OptFamilyL}
In this section we look at a specific thermodynamic transformation involving system ($S$) and catalyst ($C$) states of any dimension $m$ and $n=m^a$ respectively. For the trivial Hamiltonian where all states have same energy, the thermal state of the system is simply the fully mixed state $\frac{\mathbb{I_S}}{m}$, while any pure state corresponds to $\Pi_{\max}^S$, so we simply pick $|0\rangle\langle 0|$ without loss of generality. Note that thermo-majorization conditions are reduced to the simplest form, i.e. that 
\begin{equation}
	\omega_C\otimes \frac{1}{m}\mathbb{I}_S \rightarrow \omega_C' \otimes |0\rangle\langle 0|_S 
\end{equation}
is possible if and only if the initial state majorizes the latter, i.e.
\begin{equation}\label{eq:thermomaj_optcat}
\omega_C\otimes \frac{1}{m}\mathbb{I}_S \succ \omega_C' \otimes |0\rangle\langle 0|_S.
\end{equation} 

In this section we give a construction of catalyst states which allow this transformation, and prove that our construction achieves the optimal trace distance $d(\omega_C,\omega_C') = \frac{1}{2} \|\omega_C-\omega_C'\|_1$ in any fixed dimension $n=m^a$. Furthermore, these states are universal embezzlers, since any catalyst which successfully creates $|0\rangle\langle 0|_S$ from ${\mathbb{I}_S}/{m}$ would also allow to obtain any $\rho_S'$ from any $\rho_S$, as shown in Lemma \ref{lem:univ_emb}.

\begin{definition}
Consider integers $m\geq 2$ and $n=m^a$ where $a \geq 1$. Let $\mathcal{S}_{m,n}$ be the set of $n$-dimensional catalyst state pairs $(\omega_C, \omega_C')$ enabling the transformation 
\begin{equation}
	\omega_C\otimes \frac{1}{m}\mathbb{I}_m \rightarrow \omega_C' \otimes |0\rangle\langle 0|. 
\end{equation}
Let $d_{m,n} = \min \lbrace d(\omega_C,\omega_C')\;|\;(\omega_C ,\omega_C')\in\mathcal{S}_{m,n}\rbrace$.
\label{def:CatalystSet}
\end{definition}

\subsubsection{A family of catalyst states}\label{sec:OptFamily}
%
%
We offer the following construction of catalyst input and output states in any dimension $n=m^a$ where $m\ge 2$ and $a\ge 1$ are integers. We take the output catalyst $\omega_C' = \sum_{i=1}^n \omega'_i |i\rangle\langle i|$, where 
\begin{equation}\label{eq:finalcat}
  \omega'_1 = \frac{1}{1+(m-1)a}
  \;\text{ and }\;
  \omega'_i = \omega_1'  m^{1-  \left \lceil{\log_m i}\right \rceil }.
\end{equation}
A simple way to visualise this is as follows: for the first $m$ elements, the distribution is uniform with some probability $\omega_1$; for the next $m+1$ up to $m^2$ elements the distribution is uniform again, with probability $\omega_1/m$; and so on up to $n=m^a$. The initial $\omega_1$ is then chosen so that the full distribution is normalised.
We choose the input catalyst state to be $\omega_C = \sum_{i=1}^n \omega_i |i\rangle\langle i|$, where
\begin{equation}\label{eq:initialcat}
  \omega_i = \begin{cases}
     \omega'_1 m & \text{if $i=1$} ,\\
     \omega'_i                 & \text{if $2 \leq i \leq \frac{n}{m}$,}\\
     0                         & \text{if $i>\frac{n}{m}$}.
  \end{cases}
\end{equation}
Such a state $\omega_C$ is obtained from $\omega_C'$ by setting all the probabilities for $i>{n}/{m}$ to be zero, while renormalizing by increasing the largest peak of the probability distribution. 
Note that $\omega_1>\omega'_1$ while $\omega_i \le \omega'_i$ for all $i>2$. The trace distance between $\omega_C$ and $\omega_C'$ can be calculated to be 
\begin{equation}
  d(\omega_C,\omega_C') = 
  \frac{1}{2}\sum_{i=1}^{n} |\omega_i-\omega_i'| = 
  \sum_{i:\omega_i > _i'}( \omega_i - \omega_i' )
  = \omega_1-\omega'_1 =  \frac{m-1}{1+(m-1)a}.
\end{equation}
This shows that 
\begin{equation}\label{ineq}
	d_{m,n} \leq \frac{m-1}{1+(m-1)a}, 
\end{equation}
since we have constructed a specific state pair achieving this trace distance. In the next section we will see that for catalysts satisfying Eq.~\eqref{eq:thermomaj_optcat}, smaller values of trace distance cannot be achieved, which implies that Eq.\ (\ref{ineq}) is true with equality, and the family presented above is optimal.

\subsubsection{Optimal catalysis}
In this section we show by induction that 
\begin{equation}
	d_{m,n} \geq \frac{m-1}{1+(m-1)a}.
\end{equation}
Recall that our problem is to minimize over states $\omega_C, \omega_C'$ the trace distance $d(\omega_C,\omega_C')$ such that Eq.~\eqref{eq:thermomaj_optcat} is satisfied. 
We first show that it suffices to minimize over states which are diagonal in the same basis.

\begin{lemma}[States diagonal in the same basis]\label{lem:min_diag_ord}
Consider fixed n-tuples of eigenvalues $(\omega_1,\cdots,\omega_n)$ and $(\omega'_1,\cdots,\omega'_n)$, such that $\omega_C = \sum_i \omega_i |e_i\rangle\langle e_i|$ and $\omega_C' = \sum_i \omega'_i |f_i\rangle\langle f_i|$ are diagonal in two different bases $\lbrace | e_i\rangle\rbrace, \lbrace | f_i\rangle\rbrace$. If $(\omega_C,\omega_C')$ satisfies Eq.~\eqref{eq:thermomaj_optcat}, then there exists 
$\tilde{\omega}_C = \sum_i \tilde{\omega}_i |e_i\rangle\langle e_i|$ such that $d(\omega_C,\omega_C')\geq d(\omega_C,\tilde{\omega}_C)$ and that $(\omega_C,\tilde{\omega}_C)$ also satisfies Eq.~\eqref{eq:thermomaj_optcat}.
\end{lemma}

\begin{proof}
There are two steps in this proof: firstly, we construct $\tilde{\omega}_C$ from $\omega_C'$ and show that the trace distance decreases by invoking data processing inequality. Then, we use Schur's theorem to show that majorization holds.
Let $\tilde{\omega}_C = \mathcal{N}(\omega_C')$, where $\mathcal{N}(\rho) = \sum_{i} |e_i\rangle\langle e_i|\rho|e_i\rangle\langle e_i| $ is the fully dephasing channel in the basis $\lbrace |e_i\rangle \rbrace$. Note that since $\omega_C$ is already diagonal in $\lbrace|e_i\rangle\rbrace$, $\mathcal{N}(\omega_C)=\omega_C$. Because the trace distance is non-increasing under quantum operations~\cite{NieChu00Book}, we have
\begin{equation}
d(\omega_C,\omega_C') \geq d(\mathcal{N}(\omega_C),\mathcal{N}(\omega_C')) = d(\omega_C,\tilde{\omega}_C)\ .
\end{equation}
On the other hand, we will show that $\omega_C'\succ\tilde{\omega}_C$. For any matrix $M$, let $\lambda(M)$ be the vector of its eigenvalues.  We want to show that $\lambda(\omega'_C)\succ\lambda(\tilde{\omega}_C)$. Recall that $\tilde{\omega}_C = \mathcal{N}(\omega_C')$ and, from the definition of $\mathcal{N}$, observe that the eigenvalues $\lambda(\tilde{\omega}_C)$ are precisely the diagonal elements of $\omega_C'$ in the basis $\lbrace |e_i\rangle \rbrace$.
Schur's theorem (\cite{MarshallOlkin79}, Chapter 9, Theorem B.1.) says that for any Hermitian matrix $M$, the diagonal elements of $M$ are majorized by $\lambda(M)$.
Therefore, $\lambda(\omega'_C) \succ \lambda(\tilde{\omega}_C)$ and thus $\omega'_C \succ \tilde{\omega}_C$. 
Making use of the initial assumption $\omega_C\otimes \mathbb{I}_S /m\succ \omega_C' \otimes |0\rangle\langle 0|_S$, we now see that
\begin{equation}
\omega_C\otimes \frac{1}{m}\mathbb{I}_S \succ \omega_C' \otimes |0\rangle\langle 0|_S \succ \tilde{\omega}_C \otimes |0\rangle\langle 0|_S\ ,
\end{equation}
which concludes the proof.
\end{proof}
We are now ready to establish our lower bound on $d_{m,n}$, where will use the fact established in the previous Lemma \ref{lem:min_diag_ord} that we can take both states to be diagonal in the same basis.
\begin{theorem}
Consider integers $m\geq 2$ and $n=m^a$ where $a \geq 1$. Then 
\begin{equation}
  d_{m,n} = \frac{m-1}{1+(m-1)a},
\end{equation}
where $d_{m,n}$ is defined in Eq.~\eqref{def:CatalystSet}. Hence, the family of catalyst states from Section~\ref{def:CatalystSet} is optimal.
\end{theorem}

\begin{proof}
The majorization condition
\begin{equation}\label{eq:t3}
  \omega_C\otimes \frac{1}{m}\mathbb{I}_S \succ 
  \omega_C' \otimes |0\rangle\langle 0|_S
\end{equation}
only depends on the eigenvalues of $\omega$ and $\omega'$. For fixed eigenvalues, the trace distance $d(\omega,\omega')$ is minimized if the two states share the same eigenbasis and the eigenvalues are ordered in the same way, \textit{e.g.}, in decreasing order, as discussed in Lemma \ref{lem:min_diag_ord}. Hence, from now on we consider only diagonal states $\omega = \diag(\omega_1,\dotsc,\omega_n)$ and $\omega'=\diag(\omega'_1,\dotsc,\omega'_n)$, where $\omega_1 \ge \omega_2 \ge \dotsc \ge \omega_n$ and $\omega'_1 \ge \omega'_2 \ge \dotsc \ge \omega'_n$. Here, $\diag(\cdots)$ denotes the diagonal matrix with the corresponding diagonal elements.
To prove the theorem we only need to show that 
\begin{equation}
	d_{m,n} \ge \frac{m-1}{1+(m-1)a} 
\end{equation}
as the other inequality follows from the family of embezzling states exhibited in Section~\ref{sec:OptFamily}. We use induction on the power $a$. For the base case $a=1$, we need to show that 
$d_{m,m}\ge 1- {1}/{m}$. Consider any feasible solution $(\omega,\omega')$ in dimension $n=m$. From the majorization condition
\begin{equation}
  \left(\frac{\omega_1}{m},\dotsc,\frac{\omega_1}{m},
  \dotsc,
  \frac{\omega_m}{m},\dotsc,\frac{\omega_m}{m}\right) \succ
   (\omega'_1,\dotsc,\omega'_m,0,\dotsc,0)
\end{equation}
it follows that ${\omega_1}/{m} \ge \omega'_1$ and $\omega_i=0$ for $i\ge2$. Hence, $\omega_1=1$ and ${1}/{m} \ge \omega'_1$. Since $\omega'_1$ is the largest of the $m$ values $\omega'_i$, we get $\omega'_i = {1}/{m}$ for all $i$. Finally, a simple calculation reveals that $d(\omega,\omega') = 1-{1}/{m}$, which establishes the base case.

For the inductive step, we assume that 
\begin{equation}\label{indu}
	d_{m,n} = \frac{m-1}{1+(m-1) a}
\end{equation}	
	for some $n=m^a$ and aim to show that 
\begin{equation}	
	d_{m,k} = \frac{m-1}{1+(m-1) (a+1)}
\end{equation}	
 for $k=m^{a+1}$.
The main idea is to consider an optimal catalyst pair $(\omega,\omega')\in\mathcal{S}_{m,k}$ and from it construct a catalyst pair $(\sigma, \sigma')\in\mathcal{S}_{m,n}$ in dimension $n=m^a$. Since our construction will allow to relate $d(\sigma,\sigma')\ge d_{m,n}$ to $d(\omega,\omega')=d_{m,k}$, we then obtain a lower bound on $d_{m,k}$ in terms of $d_{m,n} $ as in Eq.\ (\ref{indu}).

Let us start by using the state pair that satisfies Eq.~\eqref{eq:t3} and achieves $d_{m,k}$, and from it derive some useful properties.
Firstly, pick $(\omega, \omega')\in\mathcal{S}_{m,k}$ so that $d(\omega,\omega')=d_{m,k}$. As before, without loss of generality, we assume that $\omega=\diag(\omega_1,\dotsc,\omega_k)$ and $\omega'=\diag(\omega_1',\dotsc,\omega'_k)$ where $\omega_1 \ge \dotsc \ge \omega_k$ and $\omega_1' \ge \dotsc \ge \omega_k'$. The majorization condition 
\begin{equation}\label{eq:maj}
  \omega\otimes \frac{1}{m}\mathbb{I}_m \succ \omega' \otimes |0\rangle\langle 0|
  \;\; \Leftrightarrow \;\;
  \left(\frac{\omega_1}{m},\dotsc,\frac{\omega_1}{m},
  \dotsc,
  \frac{\omega_k}{m},\dotsc,\frac{\omega_k}{m}\right) \succ
   (\omega'_1,\dotsc,\omega'_k,0,\dotsc,0)
\end{equation}
again implies that $\omega_1 > \omega'_1$ and $\omega_i = 0$ for $i>{k}/{m}=m^a$. To further simplify matters, we can also assume that $\omega_i \le \omega'_i$ for all $i\ge 2$. This is because we can always replace $\omega$ with $\tilde{\omega}=\diag(\tilde{\omega}_1,\dotsc,\tilde{\omega}_k)$, where
\begin{equation}
  \tilde{\omega}_i = \begin{cases}
  \omega_i' & \text{if }\omega_i> \omega'_i,\\
  \omega_i  & \text{otherwise},
  \end{cases}
\end{equation}
for $i\ge 2$ and $\tilde{\omega}_1$ is chosen so that $\sum_i \tilde{\omega}_i = 1$. In essence, all the majorization advantage of $\omega$ against $\omega'$ can be piled upon the first, largest eigenvalue of $\omega$. This replacement is valid since $(\tilde{\omega},\omega')$ still satisfies the majorization condition. Furthermore,
\begin{equation}
d(\omega,\omega')=\sum_{i:\omega_i > \omega_i'} \omega_i-\omega_i'=d(\tilde{\omega},\omega')
\end{equation} 
implies that the distance is unchanged.

Subsequently, we proceed to bound $d_{m,n}$.
To do this, construct a catalyst pair $(\sigma,\sigma')\in \mathcal{S}_{m,n}$ in dimension $n = m^a = {k}/{m}$. Essentially, this is done by directly applying a cut to the dimension of the final catalyst state $\omega'$, reducing it to having dimension ${k}/{m}=n$. Similarly, the same amount of probability is cut from the initial state, and both states are renormalized. 

Let us decribe this in more detail: denote $\delta = \sum_{i> k/m} \omega'_i$ and pick index $s$ and value $\hat{\omega}_s \le \omega_s$ so that $\sum_{i<s} \omega_i + \hat{\omega}_s =1-\delta$. Note that $s \le {k}/{m^2}$, since the majorization condition Eq.~\eqref{eq:maj} implies that 
\begin{equation}
\sum_{i\leq k/m^2} \sum_{j=1}^m \frac{\omega_i}{m} = \sum_{i\le k/m^2} \omega_i \ge \sum_{i\le k/m} \omega'_i = 1 - \delta.
\end{equation}
This inequality is obtained by summing up the first $k/m$ elements of both distributions in the L.H.S. and R.H.S. of Eq.~\eqref{eq:maj}. We now define  
\begin{align}
  \sigma &= \frac{1}{1-\delta} \diag 
  \big( \omega_1, \cdots, \omega_{s-1}, \hat{\omega}_{s},0,\cdots\cdots,0 \big),\\
  \sigma' &= \frac{1}{1-\delta} \diag 
  \big( \omega_1', \cdots, \omega_{s-1}', \omega_{s}',\omega_{s+1}',\cdots, \omega_{k/m}' \big).
\end{align}
Since $\sum_{i<s}\omega_i + \hat{\omega}_s = \sum_{i\le k/m} \omega'_i = 1 - \delta$ the states $\sigma$ and $\sigma'$ are properly normalized. To establish that $(\sigma,\sigma')\in \mathcal{S}_{m,n}$, we need to show that the majorization condition holds true. We consider two separate cases: when $\hat{\omega}_s = \omega_s$, and when $\hat{\omega}_s\neq \omega_s$.

If $\hat{\omega}_s=\omega_s$, then the inequalities in the majorization condition for $(\sigma,\sigma')$ have already been enforced by the majorization condition of $(\omega,\omega')$. Hence, $(\sigma,\sigma')$ is a valid catalyst pair in dimension 
$n={k}/{m}$, \textit{i.e.}, $(\sigma,\sigma')\in \mathcal{S}_{m,k}$. Let us now make the following two observations.
\begin{enumerate}
\item $\bm{d(\omega,\omega') \ge \delta}$. To see this, recall that $\omega_i=0$ for $i>{k}/{m}=n$, and thus
\begin{equation}
  d(\omega,\omega') = \sum_{i:\omega_i'>\omega_i} \omega_i'-\omega_i \geq \sum_{i>k/m} \omega_i = \delta.
\end{equation}
\item $\bm{d(\omega,\omega') = (1-\delta)d(\sigma,\sigma')}$. To see this, note that
\begin{equation}
  \frac{d(\omega,\omega')}{1-\delta} = \frac{1}{1-\delta}
  \sum_{i: \omega_i>\omega'_i} \omega_i-\omega'_i = \frac{\omega_1 - \omega'_1}{1-\delta}
  = d(\sigma,\sigma')
\label{eq:DistX}
\end{equation}
since only the first diagonal element of $\sigma$ is strictly larger than the corresponding diagonal element of $\sigma'$.
\end{enumerate}
Combining observations (1) and (2) gives 
\begin{equation}
  d_{m,k} = d(\omega,\omega') = (1-\delta) d(\sigma,\sigma')
  \ge  \Big[ 1-d(\omega,\omega') \Big] d(\sigma,\sigma')
  \ge (1-d_{m,k}) d_{m,n},
\end{equation}
since 
\begin{equation}
	d(\sigma,\sigma')\ge d_{m,n} = \frac{m-1}{1+(m-1)a}. 
\end{equation}
Rearranging gives us
\begin{equation}
  d_{m,k} \ge \frac{d_{m,n}}{1+d_{m,n}} = \frac{m-1}{1+(m-1)(a+1)} 
\end{equation}
and we have completed the inductive step.

If $\hat{\omega}_s\neq \omega_s$, then the majorization inequalities involving $\hat{\omega}_s$ might fail to hold. Therefore, instead of $(\sigma,\sigma')$ we consider the following, slightly different, pair of states
\begin{align}
  \zeta &=\sigma = \frac{1}{1-\delta} \diag 
  \big(\omega_1, \cdots, \omega_{s-1}, \hat{\omega}_{s},0,\cdots,0 \big),\\
  \zeta' &= \frac{1}{1-\delta} \diag 
  \big( \omega_1', \cdots, \omega_{(s-1)m}', \mathit{l},\cdots,\bar{\omega}, \omega'_{sm+1}, \cdots,\omega_{k/m}' \big),
\end{align}
where
\begin{equation}
	\mathit{l}=\frac{1}{m}\big(\omega'_{(s-1)m+1}+\dotsc+\omega'_{sm}\big).
\end{equation}
The diagonal elements of $\zeta'$ are still in descending order, and the state is properly normalized. To argue that $(\zeta,\zeta')$ is a valid pair of catalyst states, we need to verify the majorization inequalities that are not directly implied by the majorization condition for $(\omega,\omega')$. That is, we need to verify that for all $1\le j \le m$,
\begin{equation}
  C+\frac{j}{m} \hat{\omega}_s \ge C' + j\mathit{l},
\label{eq:Majcond}
\end{equation}
where $C=\sum_{i=1}^{s-1} \omega_i$ and $C'=\sum_{i=1}^{(s-1)m} \omega_i'$. 

We can see that this is true for the state pair $(\zeta,\zeta')$ because in this regime of Eq.~\eqref{eq:Majcond}, both sides increase linearly with the indices $j$, and for the endpoints $j=0$ and $j=m$, the L.H.S. is higher than the R.H.S., which is guaranteed by the majorization condition for $(\omega,\omega')$,
\begin{equation}
  C\geq C' 
  \;\; \text{and} \;\; 
  C+\hat{\omega}_s  \geq C'+ m \mathit{l}.
\end{equation}
Therefore, $(1-p)C+p(C+\hat{\omega}_s) \geq (1-p)C' + p (C'+m\mathit{l})$ for any $0\le p \le 1$. Taking $p={j}/{m}$ yields the desired inequality~(\ref{eq:Majcond}) and hence $(\zeta,\zeta')$ is a valid catalyst pair.
Lastly, note that reasoning similar to the one in Equation~\eqref{eq:DistX} can be used to deduce that 
\begin{equation}
 \frac{d(\omega,\omega')}{1-\delta} = d(\zeta,\zeta').
 \end{equation}
 Therefore, $d(\zeta,\zeta') = d(\sigma,\sigma')$ and we can use the argument from the previous case to complete the inductive step.
%
%
By this proof of induction we have shown that $d_{m,n} \geq {m-1}/({1+(m-1)a})$ for all $m, n=m^a$ and $a\geq 1$. This together with the conclusion in Section \ref{sec:OptFamily} that 
$d_{m,n}\leq {m-1}/({1+(m-1)a})$ proves that 
\begin{equation}
d_{m,n}=\frac{m-1}{1+(m-1)a},
\end{equation}
and the state pair described in Eq.~\eqref{eq:finalcat} and \eqref{eq:initialcat} is optimal.
\end{proof}

\section{Limits of thermal embezzling from constraints on dimension}\label{sec:DimBounds}
\subsection{Diagonal states}\label{subsec:thermalcat_diag}
In our work, we use two particular quantities, which are the R{\'e}nyi divergences for $\alpha={1}/{2}$ and $\alpha =\infty$, which for classical probability distributions have the following form:
\begin{eqnarray}\label{eq:dhalf&dmin}
&D_{1/2}(p \| q)= -2 \log \sum_i \sqrt{p_i q_i}, \quad
& D_\infty(p \| q) = \lim_{\alpha\to \infty}D_\alpha(p \| q)=\log \max_i \frac{p_i}{q_i}.
\end{eqnarray}
As mentioned in Section \ref{subsec:d_alpha}, given Hamiltonians $H_S$ and $H_C$, it suffices to consider
\begin{equation}
\omega_C \otimes \tau_S \rightarrow \omega'_C \otimes \Pi^S_{\rm max}.
\end{equation}
Here, we prove whenever the dimension of the catalyst (and system) are finite, there exists a lower bound on the accuracy of thermal embezzling. Such a bound is dependent on $H_S$ and $H_C$. To do so, consider the problem 
\begin{eqnarray}\label{eq:embezz_prob}
\varepsilon = &&\min \quad \frac{1}{2}\|\omega_C-\omega'_C\|_1 \\
&&{\rm s.t.}\quad ~\omega_C \otimes \tau_S \rightarrow \omega'_C \otimes \Pi^S_{\rm max},~ 0\leq \omega,\sigma\leq \mathcal{\id}.\nonumber
\end{eqnarray}

In Ref.~\cite{secondlaw}, it has been shown that for initial and target states commuting with the Hamiltonian $H_S$, it is sufficient to consider catalyt states commuting with $H_C$. Therefore, since $\tau_S$ and $\Pi_{\max}^S$ both commute with $H_S$, it is sufficient to consider input and output catalysts states which are diagonal in the basis of $H_C$. Since all $\alpha$ R{\'e}nyi divergences are thermal monotones according to Lemma \ref{lem:secondlaw}, in particular the min-relative entropy ($D_\infty$), for $\alpha\rightarrow\infty$, 
\begin{equation}
	D_\infty (\rho\|\rho')= \max_i\log \frac{\rho_i}{\rho_i'}
\end{equation}
where $\rho_i$ and $\rho'_i$ are the eigenvalues of $\rho,\rho'$ respectively. Therefore, satisfying the thermo-majorization conditions in Eq. \eqref{eq:embezz_prob} implies that
\begin{align*}
D_\infty (\omega_C\otimes\tau_S\|\tau_{CS}) &\geq D_\infty (\omega'_C\otimes \Pi^S_{\rm max}\|\tau_{CS}).
\end{align*}
To further simplify this expression, note that $\tau_{CS}=\tau_C\otimes\tau_S$ and that $D_\alpha (\rho\otimes\rho'\|\sigma\otimes\sigma')=D_\alpha(\rho\|\sigma)+D_\alpha(\rho'\|\sigma')$. The additivity of R{\'e}nyi divergences under tensor products holds for all states. Furthermore, $D_\alpha (\rho\|\rho)=0$ for any $\rho$. Therefore, we arrive at the expression
\begin{equation}\label{eq:dalpha_sep}
D_\infty (\omega_C\|\tau_C) + 0 \geq D_\infty (\omega'_C\|\tau_C) + \log \frac{Z_S}{e^{-\beta E^S_{\rm max}}},
\end{equation}
where $Z_S$ is the partition function of the system. The spectral values of $\omega_C$ and $\omega'_C$ are denoted as
$\{\omega_j\}$ and $\{\omega'_j\}$, respectively.
Using the definition of $D_\infty$ as shown in Eq.~\eqref{eq:dhalf&dmin}, we obtain
\begin{equation*}
\max_i \frac{\omega_i}{\tau_i}  \geq  \frac{Z_S}{e^{-\beta E^S_{\rm max}}}
\max_j \frac{\omega'_j}{\tau_j},
\end{equation*}
where 
\vspace{-0.5cm}
\begin{equation}
	\tau_j = \frac{e^{-\beta E_j^{C}}}{Z_C} 
\end{equation}
are the eigenvalues of the thermal state for the catalyst, for the energy eigenstate with energy eigenvalue $E_i^C$, with normalization $Z_C$, the partition function of the catalyst. 
Since $\hat{\varepsilon}$ is the minimum trace distance between states $\omega_C, \omega'_C$, and $D_\infty$ depends only on the maximum of $\omega'_i/\tau_i$ across the distribution, the optimal strategy to increase $D_\infty$ while going from $\omega'_C$ to $\omega_C$ is to increase a specific $\omega'_i$ by an amount $\hat{\varepsilon}$.
Therefore, we can consider a relaxation of Eq.~\eqref{eq:embezz_prob}
\begin{eqnarray}
\hat{\varepsilon} = &&\min \quad \frac{1}{2}\|\omega_C-\omega'_C\|_1\\
&&{\rm s.t.}\quad ~ \max_i \frac{\omega'_i+\hat{\varepsilon}}{\tau_i} \geq   \frac{Z_S}{e^{-\beta E^S_{\rm max}}}
\max_j \frac{\omega'_j}{\tau_j}  , \\
&&\qquad\quad \forall j, 0 < \omega'_j \leq 1.
\end{eqnarray}

In the next lemma, we show that 
$\varepsilon \geq \hat{\varepsilon} \geq \delta > 0$ whenever $E^{C}_{\rm max}, E^{S}_{\rm max} <\infty$.

\begin{lemma}[Lower bound to error in catalysis]\label{lem:lb_err_catalysis}
Consider system and catalyst Hamiltonians which are finite-dimensional, and denote $\lbrace E_i^S \rbrace_{i=1}^m$, $\lbrace E_i^C \rbrace_{i=1}^n$ to be the set of energy eigenvalues respectively. Then for some fixed $E_{\rm max}^{C}, E^{S}_{\rm max}$, consider any probability distribution $r$ (which corresponds to eigenvalues of a catalyst $\omega$), and $\hat{\varepsilon}$ such that
\begin{equation}
	\max_i \frac{r_i+\hat{\varepsilon}}{\tau_i} \geq 
	\frac{Z_S}{e^{-\beta E^S_{\rm max}}}\max_j \frac{r_j}{\tau_j} , 
	~ \forall j, 0 < r_j \leq 1,
\end{equation}
where $\tau_i = {e^{-\beta E_i^{C}}}/{Z_C}$. Note that index $i$ runs over all energy levels $E_i^C$. Then 
\begin{equation}
	\hat{\varepsilon} \geq \left(\frac{Z_S}{e^{-\beta E^S_{\rm max}}}-1\right)\frac{e^{-\beta E^C_{\rm max}}}{Z_C} \neq 0.
\end{equation}
In other words, thermal embezzling of diagonal states with arbitrary accuracy is not possible.
\end{lemma}
\begin{proof}
Firstly, let $r^*, \tau^*$ indicate the pair such that ${r^*}/{\tau^*}=\max_j {r_j}/{\tau_j}$. Then 
\begin{align*}
\max_i \frac{r_i}{\tau_i} + \max_i \frac{\hat{\varepsilon}}{\tau_i} & \geq \max_i \frac{r_i + \hat{\varepsilon}}{\tau_i} \geq \frac{r^*}{\tau^*} \frac{Z_S}{e^{-\beta E^S_{\rm max}}}.
\end{align*}
The first term of L.H.S. is equal to $r^*/\tau^*$, and therefore can be grouped with the R.H.S. to form
\begin{align*}
\max_i \frac{\hat{\varepsilon}}{\tau_i} &\geq \frac{r^*}{\tau^*} \left(\frac{Z_S}{e^{-\beta E^S_{\rm max}}} -1 \right) \geq \frac{Z_S}{e^{-\beta E^S_{\rm max}}} -1,
\end{align*}
since we know that $D_\infty (r\|q)=\log\max_i r_i/\tau_i = \log r^*/\tau^* \geq 0$, therefore $r^*/\tau^*\geq 1$. Finally, taking the maximization of $1/\tau_i$ over $i$ gives 1/$\tau_{\min}$, recall that $\tau_i$ corresponds to probabilities of the thermal state being in the eigenstate with energy $E_i$. Therefore, $\tau_{\min} = e^{-\beta E_{\max}^C}/Z_C$, and we get
\begin{equation}
\hat{\varepsilon} \geq \left(\frac{Z_S}{e^{-\beta E^S_{\rm max}}} -1 \right) \frac{e^{-\beta E^C_{\rm max}}}{Z_C}.
\end{equation}
\end{proof}
\subsection{Arbitrary states}\label{subsec:thermalcat_arb}
The case of arbitrary states are treated separately, since our Lemma \ref{lem:univ_emb} on universal embezzlers hold only for diagonal states, where necessary and sufficient conditions are known for state transformations. Nevertheless, since the monotonicity of $D_\alpha$ is necessary for arbitrary state transformations $\rho_S\rightarrow\rho_S'$, one can use techniques very similar to those in Section \ref{subsec:thermalcat_diag} to lower bound the embezzling error, if we minimize over diagonal catalysts.

More precisely, denote $\varepsilon (\rho_S,\rho_S')$ to be the solution of
\begin{equation}\label{eq:arbstatedalpha}
\begin{aligned}
\min  &\quad\frac{1}{2}\|\omega_C-\omega'_C\|_1 \\
{\rm s.t.} &\quad D_\infty (\omega_C\otimes\rho_S\|\tau_{CS}) \geq D_\infty (\omega'_C\otimes \rho_S'\|\tau_{CS}) ,~ 0\leq \omega,\sigma\leq \mathcal{\id}.
\end{aligned} 
\end{equation}

Recall that $\tau_{CS}=\tau_C\otimes\tau_S$, and that $D_\alpha$ is additive under tensor products. Therefore, by defining
\begin{equation}
\kappa_1 (\rho_S,\rho_S'):=D_\infty(\rho_S'\|\tau_S)-D_\infty(\rho_S\|\tau_S),
\end{equation}
we can rearrange the first constraint in Eq.~\eqref{eq:arbstatedalpha} 
\begin{equation}
D_\infty (\omega_C\|\tau_{C}) \geq D_\infty (\omega'_C|\tau_{C})+ \kappa_1(\rho_S,\rho_S').
\end{equation} 
Note that this is almost equivalent to Eq.~\eqref{eq:dalpha_sep}, except the constant $\log Z_S/e^{-\beta E_{\rm max}^S}$ previously is now replaced with $\kappa_1(\rho_S,\rho_S')$. By following the same steps used to prove Lemma \ref{lem:lb_err_catalysis}, we obtain a lower bound depending on $\rho_S, \rho_S'$.

\begin{lemma}\label{lem_err_catalysis_arb}
Consider system and catalyst Hamiltonians which are finite-dimensional, and denote $\lbrace E_i^S \rbrace_{i=1}^m$ and $\lbrace E_i^C \rbrace_{i=1}^n$ to be the set of energy eigenvalues respectively. Then for some fixed $ 0\leq E_{\rm max}^{C}, E^{S}_{\rm max}$, consider any probability distribution $r$ (which corresponds to eigenvalues of a catalyst $\omega$), and $\hat{\varepsilon}$ such that
\begin{equation}
	\max_i \frac{r_i+\hat{\varepsilon}}{\tau_i} \geq 
	2^{\kappa_1(\rho_S,\rho_S')} \cdot \max_j \frac{r_j}{\tau_j} , 
	~ \forall j, 0 < r_j \leq 1,
\end{equation}
where $\tau_i = {e^{-\beta E_i^{C}}}/{Z_C}$ and $\kappa_1 (\rho_S,\rho_S')=D_\alpha(\rho_S'\|\tau_S)-D_\alpha(\rho_S\|\tau_S)$. Note that index $i$ runs over all energy levels $E_i^C$. Then 
\begin{equation}
	\hat{\varepsilon} \geq \left[2^{\kappa_1(\rho_S,\rho_S')}-1\right]\frac{e^{-\beta E^C_{\rm max}}}{Z_C} \neq 0.
\end{equation}
This implies thermal embezzling with arbitrary accuracy, using a diagonal catalyst is not possible.
\end{lemma}

Comparing Lemma \ref{lem:lb_err_catalysis} and Lemma \ref{lem_err_catalysis_arb} which are very similar, one sees that for non-diagonal states Lemma \ref{lem_err_catalysis_arb} gives a state-dependent lower bound on the embezzling error. However for diagonal states, the bound in Lemma \ref{lem:lb_err_catalysis} can be made state-independent because of the existence of universal embezzlers. 

\subsection{Relation to energy constraints}
Rather than bounding the dimension of the catalyst, one can ask if restrictions on other physical quantities such as the average energy of the catalyst would prevent indefinitely accurate embezzling from occurring. While this by itself is an independently interesting problem, we can first note that such restrictions are sometimes related to restrictions on the dimension. In one direction this is straightforward: if the catalyst is finite-dimensional, then the average energy and all other moments of energy distribution would be finite as well. 

Here, we show that by restricting the first and second moments of the energy distribution of the catalyst to be finite, this implies that the states involved are always close to finite-dimensional states. In other words, if we consider the set of catalysts such that the average and variance of energy is finite, then for any such catalyst state from this set, there always exists a finite-dimensional state $\varepsilon$-close to it. This can be shown by invoking a simple theorem, namely the Chebyshev inequality which says that for given any finite non-zero error $\varepsilon$, the support of the energy distribution must be finite.

\begin{lemma}[Chebyshev inequality]\label{lem:chebyshev}
Consider a random variable $X$ with finite mean $\bar{X}$ and finite variance $\sigma_X^2$, then for all $k>0$,
\begin{equation}
\pp [|X-\bar{X}|\geq k	] \leq \frac{\sigma_X^2}{k^2}.
\end{equation}
\end{lemma}

\begin{theorem}[Chebyshev inequality applied to energy distributions]
Consider a probability distribution $p$ over some non-degenerate energy values $E$, where both mean $\bar{E}=\langle E\rangle$, and variance $ \sigma_E^2=\langle [E-\bar{E}]^2\rangle$ are finite. Then for any $\varepsilon > 0$, there exists $E_{\rm max} < \infty$ such that $\pp [E\geq E_{\rm max}] \leq \varepsilon$.
\end{theorem}
\begin{proof}
For any $\varepsilon >0$, let some $k={\sigma_E}/{\sqrt{\varepsilon}}$. Denote $E_{\rm max}= \bar{E} + k$. Then by Lemma~\ref{lem:chebyshev}, 
\begin{equation}
\pp [E\geq E_{\rm max}] \leq \pp [|E-\bar{E}|\geq k] \leq \varepsilon.
\end{equation}
\end{proof}

\section{Limits of thermal embezzling from energy constraints}\label{sec:ELB}
In this section we provide lower bounds for the error in catalysis, given constraints on the average energy of the catalyst state. We do so by adding a constraint on the average energy of the catalyst to the problem stated in Eq.~\eqref{eq:embezz_prob}. By looking at the R{\'e}nyi divergence for $\alpha = {1}/{2}$, we can show a non-zero lower bound on the catalytic error, for cases where the partition function of the catalyst Hamiltonian $Z_C$ is finite. This minimal assumption covers most physical scenarios, especially if we want the thermal state to be a trace class operator to begin with. Again we start with diagonal states, then later generalize to arbitrary states.

\subsection{Diagonal states}\label{subsec:diag_ELB}
Firstly, let us recall the problem stated in Eq.\ \eqref{eq:embezz_prob}. We aim at minimizing the trace distance between all input and output catalyst states, such that the most significant
 thermal embezzlement of a smaller system $S$ can be achieved. We denote again the initial and final catalysts by $\omega_C$ and $\omega'_C$
 with spectral values $\{\omega_j\}$ and $\{\omega'_j\}$.
 Again, by restricting ourselves to look at catalyst diagonal in the Hamiltonian basis, and by invoking only the thermal monotone $D_{1/2} (.\|.)$, one can find the alternative relaxed problem
\begin{eqnarray}\label{eqarray1}
	 \text{min} && \quad\frac{1}{2}\sum_{j=1}^\infty |\omega_j -\omega'_j|,\\
	\text{s.t.} &&\quad
	\sum_{j=1}^\infty ({\omega'_j}^{1/2}- A^{1/2}  \omega_j^{1/2} ) \gamma^{E_j^C}\geq 0,\sum_{j=1}^\infty \omega'_j=1,\,
	\sum_{j=1}^\infty \omega_j=1,\nonumber\\
	&&\quad\omega'_j, \omega_j\geq 0 \,\,\,\forall j,~{\rm and}~~\sum_{j=1}^\infty E_j^C \omega_j \leq E ,\nonumber
\end{eqnarray}
where 
\begin{equation}\label{eq:A}
	A = \frac{Z_S}{e^{-\beta E^S_{\rm max}}} 
\end{equation}
and $\gamma = e^{-{\beta}/{2}} < 1$. Furthermore, since $A = 1/\min_i \tau_i $ with $\tau_i$ forming a probability distribution (that of a thermal state), one can deduce that whenever the dimension of system $S$ is $m\geq 2$, $A\geq m \geq 2$ holds as well. 

The solution of this minimization problem serves as a lower bound to the optimal trace distance error. This problem can be relaxed to a convex optimisation problem. We can arrive at a simple bound, however, with rather non-technical means. In essence, we introduce split bounds, so that the optimization can be written as two independent, individually significantly simpler optimization problems. We make use
of the inequality
\begin{equation}
	x^{1/2} - a^{1/2} y^{1/2} \leq |x-y|^{1/2} - f(a)y,
\end{equation}	
which holds true for
$x,y\in[0,1], a\geq 2$ and with $f:\rr^+\rightarrow \rr^+$ defined as	
\begin{equation}\label{eq:f(a)}
	f(a)=\frac{1}{2} \frac{a^2}{a^2+1}. 
\end{equation}
We can then relax the problem by replacing the first constraint in Eq.\ \eqref{eqarray1}, with $x_j$ taking the role of $|\omega_j-\omega'_j|$, to arrive at
\begin{eqnarray}
	 \text{min} \quad && \frac{1}{2}\sum_{j=1}^\infty x_j,\\
	\text{s.t.} \quad && \sum_{j=1}^\infty \left[x_j^{1/2} -   f(A) \omega_j \right] e^{-\beta E_j^C/2}\geq 0, \sum_{j=1}^\infty \omega_j=1,\nonumber\\
	&&x_j, \omega_j\geq 0 \,\,\,\forall j,~{\rm and}~\sum_{j=1}^\infty E_j^C \omega_j \leq E .\nonumber
\end{eqnarray}
These are now two independent optimisation problems, by treating $x_j$ and $\omega_j$ as independent variables. Define $\varepsilon_C$ to be the solution of the simple linear problem involving 
only variables $\{\omega_j\}$, which we explicitly write out in Corollary \ref{cor:varepsC_lb}. In this subproblem, one notes that the constraint on expectation value of the energy implies that the total probability of having relatively low energy eigenvalues cannot be vanishingly small, which we prove in Lemma \ref{lem:EvsW}. One can then use this fact to place a lower bound on the quantity $\varepsilon_C$, which we detail in Corollary \ref{cor:varepsC_lb}.

\begin{lemma}[Lower bound to sums of eigenvalues]\label{lem:EvsW} 
Consider any probability distribution $\lbrace \omega_i \rbrace$ over ascendingly ordered energy eigenvalues $\lbrace E_i^C \rbrace$, with the property 
that the energy eigenvalues are unbounded, i.e. $\lim_{n\rightarrow\infty} E^C_n = \infty$. If the expectation value of energy $\sum_{i=1}^\infty \omega_i E_i^C \leq E$ for some finite constant $E$, 
define for any $0<W<1$
\begin{equation}
	j(W)= \min \biggl\{ j: E_{j+1}^C > \frac{E}{1-W} \biggr\}.
\end{equation}
Then 
\begin{equation}
\displaystyle\sum_{i=1}^{j(W)} \omega_i\geq W.
\end{equation}
\end{lemma}
\begin{proof}
One can easily prove this by contradiction. Assume that 
\begin{equation}	
	\displaystyle\sum_{i=1}^{j(W)} \omega_i < W 
	\end{equation}
	and therefore $\sum_{i=j(W)+1}^\infty \omega_i > 1-W$. This violates the energy constraint, since 
	\begin{equation}
	\displaystyle\sum_{i=j(W)+1}^\infty \omega_i E_i^C > (1-W)\frac{E}{1-W} = E.
\end{equation}
\end{proof}

\begin{corollary}[Lower bound to $ \varepsilon_C$]\label{cor:varepsC_lb}
For a set of unbounded energy eigenvalues $\lbrace E_i^C \rbrace$, consider the minimization problem
\begin{eqnarray*}
	  \varepsilon_C = \mathrm{min} &&~\sum_{j=1}^\infty \omega_j e^{-\beta E_j^C},\\
	{\rm s.t.} \,\,&&~\sum_{j=1}^\infty \omega_j=1, ~\omega_j\geq 0 ~ \forall j,~{\rm and}~\sum_{j=1}^\infty E_j^C \omega_j \leq E .
\end{eqnarray*}
Denote $\gamma = e^{-\beta} \in (0,1)$. Then for $j(W)= {\rm min} \lbrace j: E_{j+1}^C > {E}/{1-W}\rbrace$,
\begin{eqnarray}
\varepsilon_C \geq \max_{W\in(0,1)} W\gamma^{E_{j(W)}}.
\end{eqnarray}
\end{corollary}
\begin{proof}
This is a direct application of Lemma \ref{lem:EvsW}, since the first and second constraints are satisfied automatically by any probability distribution. Given some $W\in(0,1)$, by Lemma \ref{lem:EvsW} we know that $\sum_{i=1}^{j(W)} \omega_i \geq W$. The objective function then can be lower bounded as
\begin{equation}
\sum_{i=1}^\infty \omega_i e^{-\beta E_i} \geq \sum_{i=1}^{j(W)} \omega_i e^{-\beta E^C_{j(W)}} \geq W \gamma^{E^C_{j(W)}},
\end{equation}
for any such $W$. To obtain the best lower bound, one maximizes over all $W\in(0,1)$.
\end{proof}

\begin{remark}[Temperature independence]
The bound obtained in Corollary \ref{cor:varepsC_lb} is dependent on temperature of the bath, and goes to zero in the limit $T\rightarrow 0$.
\end{remark}

We have now solved the subproblem involving variables $\{\omega_i\}$. Inserting the solution into the former optimisation problem,
we arrive at the lower bound for $\varepsilon$,
\begin{eqnarray*}
	 \text{min}~ \frac{1}{2} \sum_{j=1}^\infty x_j\quad\text{s.t.} 
	\sum_{j=1}^\infty x_j^{1/2}  e^{-\beta E_j^C/2}\geq f(A) \varepsilon_C,~x_j\geq 0 \,\,\,\forall j.
\end{eqnarray*}
The optimal solution for this minimization can easily be lower bounded by considering the Lagrange dual, which is 
\begin{eqnarray*}
	 \text{max} ~-\frac{1}{4}  \lambda^2\sum_{j=1}^\infty e^{-\beta E_j^C}  +\lambda f(A) \varepsilon_C,\quad\text{s. t. } & \lambda\geq 0.
\end{eqnarray*}
In fact, this can obviously be immediately solved as a quadratic problem in one variable. Let 
\begin{equation}
	g(\lambda)=\sum_{j=1}^\infty e^{-\beta E_j^C} \lambda^2 +\lambda \varepsilon_C, 
\end{equation}
and consider the stationary point of the function by setting first derivative w.r.t. $\lambda$ to zero,
\begin{equation}
-\frac{1}{2} \lambda \sum_i e^{-\beta E_i} + f(A)\varepsilon_C =0,
\end{equation}
where the second derivative is negative, hence this implies a maximum point. Substituting this into the objective function gives ${f(A)\varepsilon_C^2}/{Z_C}$, and hence we conclude that
\begin{equation*}
	\varepsilon \geq \frac{1}{2}\frac{f(A)^2\varepsilon_C^2}{Z_C}.
\end{equation*}

In this way, we arrive at the main result.
\begin{theorem}[Energy constraint limits the accuracy of thermal catalysis]\label{th:ELB}
Consider the transformation $\omega_C\otimes\tau_S\rightarrow\omega_C'\otimes|E^S_{\rm max}\rangle\langle E^S_{\rm max}|$, where $d_{\rm opt}=\frac{1}{2}\|\omega_C-\omega_C'\|_1= \frac{1}{2}\varepsilon$ is the error induced on the catalyst. Then for all catalyst states with finite average energy, $d_{\rm opt}$ is lower bounded by
\begin{eqnarray*}
	d_{\rm opt} &\geq& \frac{1}{2} \frac{f(A)^2\varepsilon_C^2}{Z_C},
\end{eqnarray*}
where $f(x)$ is defined in Eq.~\eqref{eq:f(a)}, $A = {Z_S}/{e^{-\beta E^S_{\rm max}}} $, $\varepsilon_C  = \max_{W\in(0,1)} W\gamma^{E^C_{j(W)}}$ and $j(W) = \min \lbrace j: E^C_{j+1} > {E}/(1-W)\rbrace$.\\
In other words, thermal embezzling of diagonal states with arbitrary accuracy is not possible.
\end{theorem}

\subsection{Arbitrary states}
Similar to our previous discussions in Section \ref{subsec:thermalcat_arb} , when the states $\rho_S$ or $\rho_S'$ are non-diagonal, we can still obtain a state dependent lower bound for the embezzling error. For any state $\rho_S,~\rho_S'$, let us define the quantity
\begin{equation}
\kappa_2 (\rho_S,\rho_S'):=D_{1/2}(\rho_S'\|\tau_S)-D_{1/2}(\rho_S\|\tau_S).
\end{equation}
Then a lower bound can be obtained by following the steps as proved in Section \ref{subsec:diag_ELB}, only now replacing the constant $A$ defined in Eq.~\eqref{eq:A} with a state-dependant function. 

\begin{lemma}
For arbitrary states $\rho_S$ and $\rho_S'$, consider the transformation $\omega_C\otimes\rho_S\rightarrow\omega_C'\otimes\rho_S'$, where $d_{\rm opt}=\frac{1}{2}\|\omega_C-\omega_C'\|_1= \frac{1}{2}\varepsilon$ is the error induced on the catalyst. Then for all catalyst states with finite average energy, $d_{\rm opt}$ is lower bounded by
\begin{eqnarray*}
	d_{\rm opt} &\geq& \frac{1}{2} \frac{f(2^{\kappa_2(\rho_S,\rho_S')})^2\varepsilon_C^2}{Z_C},
\end{eqnarray*}
where $f(x)$ is defined in Eq.~\eqref{eq:f(a)}, $\kappa_2 (\rho_S,\rho_S')=D_{1/2}(\rho_S'\|\tau_S)-D_{1/2}(\rho_S\|\tau_S)$, $\varepsilon_C  = \max_{W\in(0,1)} W\gamma^{E^C_{j(W)}}$ and $j(W) = \min \lbrace j: E^C_{j+1} > {E}/{1-W}\rbrace$. This implies that thermal embezzling with arbitrary accuracy using a diagonal catalyst is not possible.
\end{lemma}

\end{document}